\RequirePackage{snapshot}
\documentclass[runningheads,orivec]{llncs}
\usepackage[utf8]{inputenc}
\usepackage{graphicx}
\usepackage{xspace}
\usepackage{amsfonts}
\usepackage{mathtools}
\usepackage{wrapfig}
\usepackage{color}
\usepackage{listings}
\usepackage{xcolor}
\usepackage{cite}
\usepackage[textsize=tiny]{todonotes}
\usepackage{rotating}
\usepackage{array}
\usepackage[export]{adjustbox}
\PassOptionsToPackage{hyphens}{url}
\usepackage[colorlinks,urlcolor=blue,linkcolor=magenta,citecolor=red,
linktocpage=true]{hyperref}

\long \def\ignore#1{\relax}

\newcommand{\Nodes}{\ensuremath{\mathcal{N}}}

\newcommand{\GetConf}{\ensuremath{.\mymathtt{conf}}}
\newcommand{\GetCost}{\ensuremath{.\mymathtt{cost}}}
\newcommand\fun{\rightarrow}
\newcommand\dom{\mymathtt{dom}}
\newcommand\ran{\mymathtt{ran}}

\newcommand{\CompTy}[1]{\ensuremath{\mathcal{#1}}}
\newcommand{\Plan}[1]{\ensuremath{\mathsf{#1}}}

\newcommand{\AEOLUS}{Aeolus\xspace}

\newcommand{\trans}[1]{\ensuremath{\xrightarrow{#1}}}

\newcommand{\Ifaces}{\ensuremath{\mathcal{I}}}
\newcommand{\Resources}{\ensuremath{\mathcal{Z}}}
\newcommand{\Res}{\ensuremath{\mathcal{R}}}
\newcommand{\Actions}{\ensuremath{\mathcal{A}}}
\newcommand{\ResTypesF}{\ensuremath{{\Gamma}}}
\newcommand{\ResTy}[1]{\ensuremath{\mathcal{#1}}}
\newcommand{\Conf}[1]{\ensuremath{\mathcal{#1}}}

\newcommand{\mymathtt}[1]{\mbox{\footnotesize\tt #1}}

\newcommand{\GetProv}{\ensuremath{.\mymathtt{prov}}}

\newcommand{\GetSReq}{\ensuremath{.\mymathtt{req$_{\mymathtt{s}}$}}}
\newcommand{\GetWReq}{\ensuremath{.\mymathtt{req$_{\mymathtt{w}}$}}}
\newcommand{\GetResource}{\ensuremath{.\mymathtt{res}}}

\newcommand\pfun{\mathrel{\ooalign{\hfil$\mapstochar\mkern5mu$\hfil\cr$\to$\cr}}}

\newcommand{\tooltxt}{SmartDepl\xspace}
\newcommand{\tool}{{\sf \tooltxt}\xspace}

\lstdefinestyle{ANTLR}{
    basicstyle=\footnotesize\ttfamily\color{black},    breaklines=true,    moredelim=[s][\color{magenta}\ttfamily]{'}{'},    moredelim=*[s][\color{black}\ttfamily]{options}{\}},    commentstyle={\color{gray}\itshape},    morecomment=[l]{//},    emph={        INT,        VARIABLE,
        ID,
        RE
        },emphstyle={\color{blue}\ttfamily},    alsoletter={:,|,;},    morekeywords={:,|,;},    keywordstyle={\color{black}},}

\lstdefinestyle{json}{
    basicstyle=\footnotesize\ttfamily,
    showstringspaces=false,
    breaklines=true,
    moredelim=[s][\color{magenta}\ttfamily]{"}{"},
    morekeywords={:,[,],\{\}},
    literate=
     *{0}{{{\color{blue}0}}}{1}
      {1}{{{\color{blue}1}}}{1}
      {2}{{{\color{blue}2}}}{1}
      {3}{{{\color{blue}3}}}{1}
      {4}{{{\color{blue}4}}}{1}
      {5}{{{\color{blue}5}}}{1}
      {6}{{{\color{blue}6}}}{1}
      {7}{{{\color{blue}7}}}{1}
      {8}{{{\color{blue}8}}}{1}
      {9}{{{\color{blue}9}}}{1},
}

\lstdefinelanguage{ABS}{keywords=
{null,this,thisDC,dyndelta,new,data,type,def,case,of,cog,class,interface,extends
,implements,if,else,await,get,total,load,transfer,movecogto,Fut,return,skip,
while,module,duration,now,deadline,import, export, uses, from, destiny, 
suspend,delta,adds,modifies,removes,original,productline,features,core,
corefeatures,optionalfeatures,after,when,product,hasAttribute,hasMethod,root,
extension,group,allof,oneof,require,exclude,original,ifin,ifout,opt}, 
sensitive=true, comment=[l]{//}, morecomment=[s]{/*}{*/}, morestring=[b]"}

\lstdefinestyle{absstyle}{
language=ABS,
columns=fullflexible,
mathescape=true,showstringspaces=false,keywordstyle=\bf\sffamily,
commentstyle=\sl\sffamily,basicstyle=\footnotesize\sffamily,
inputencoding=latin1, literate=*,
extendedchars}

\lstdefinelanguage{MYSPEC}{keywords=
{sum,and,or,impl,forall,in,of,type,used,by,exists}, 
sensitive=true,morecomment=[s]{'}{'}}

\lstdefinestyle{declangstyle}{
language=MYSPEC,
columns=fullflexible,
mathescape=true,showstringspaces=false,keywordstyle=\bf\ttfamily,
commentstyle=\sl\ttfamily\color{magenta},basicstyle=\footnotesize\ttfamily,
inputencoding=latin1,
literate=*,
extendedchars}

\begin{document}
\title{\texorpdfstring{Optimal and Automated\\ Deployment for Microservices}{Optimal and Automated Deployment for Microservices}}
\author{
Mario Bravetti \inst{1} \and
Saverio Giallorenzo \inst{2} \and
Jacopo Mauro \inst{2} \and \\
Iacopo Talevi \inst{1} \and
Gianluigi Zavattaro \inst{1} }
\authorrunning{M. Bravetti et al.}
\institute{FOCUS Research Team, University of Bologna/INRIA, Italy
\and
Southern Denmark University, Denmark}
\maketitle              \begin{abstract}
Microservices are highly modular and scalable Service Oriented Architectures.
They underpin automated deployment practices like Continuous Deployment and
Autoscaling.
In this paper we formalize these practices and show that automated deployment
--- proven undecidable in the general case --- is algorithmically treatable
for microservices.
Our key assumption is that the configuration life-cycle of a microservice is
split in two phases:
(i) creation, which entails establishing initial connections with already
available microservices, and (ii) subsequent binding/unbinding with other
microservices.
To illustrate the applicability of our approach, we implement an automatic
optimal deployment tool and compute deployment plans for a realistic
microservice architecture, modeled in the Abstract Behavioral Specification
(ABS) language.
\end{abstract}

\section{Introduction}

Inspired by service-oriented computing, Microservices structure software
applications as highly modular and scalable compositions of fine-grained and
loosely-coupled services~\cite{DGLMMMS17}. These features support modern
software engineering practices, like continuous
delivery/deployment~\cite{continuous_delivery} and application
autoscaling~\cite{autoscaling}. Currently, these practices focus on single
microservices and do not take advantage of the information on the
interdependencies within an architecture. On the contrary, architecture-level
deployment plans can \emph{i}) optimize global scaling --- e.g., avoiding the
overhead of redundantly detecting inbound traffic and sequentially scale each
microservice in a pipeline --- and \emph{ii}) avoid ``domino'' effects due to
unstructured scaling --- e.g., cascading slowdowns or outages~
\cite{hellerstein2018serverless,cascadingFailure,outagesDowntime}.

In this paper, we formally investigate the problem of automatizing the
deployment and reconfiguration (e.g., horizontal or vertical scaling) of
microservice architectures, proving formal properties and
presenting an implemented solution.

In our work, we follow the approach taken by the 
\emph{\AEOLUS component model}~\cite{sefm-aeolus,infCom14,concur15},
which was used to formally define the problem of
deploying component-based software systems and to prove that,
in the general case, such problem is undecidable~\cite{sefm-aeolus}.
The basic idea of \AEOLUS is to enrich the specification
of components with a finite state automaton 
that describes their deployment life cycle.
Previous work identified decidable fragments of the \AEOLUS model:
e.g., removing from \AEOLUS replication constraints (e.g., used to specify a minimal amount of services connected to a load balancer)
makes the deployment problem decidable, but non-primitive 
recursive~\cite{infCom14};
removing also 
conflicts (e.g., used to express the impossibility to 
deploy in the same system two types of components)  
makes the problem PSpace-complete~\cite{mfcs15}
or even poly-time~\cite{sefm-aeolus}, but under the
assumption that every required component can be (re)deployed from scratch.

Our intuition is that the \AEOLUS model can be adapted to formally
reason on the deployment of microservices. To achieve our goal, we
significantly revisit the formalization of the deployment problem,
replacing \AEOLUS components with a model of \emph{microservices}.
The main difference between our model of microservices and \AEOLUS
components lies in the specification of their deployment life cycle.
Here, 
instead of using the full power of finite state automata
(like in \AEOLUS and other TOSCA-compliant deployment models~\cite{Brogi15}), 
we assume microservices to have two states: (i) creation and (ii)
binding/unbinding. Concerning creation, we use \emph{strong} dependencies
to express which microservices must be immediately connected to newly
created ones. After creation, we use \emph{weak} dependencies to indicate
additional microservices that can be
bound/unbound.
The principle that guided this modification comes from state-of-the-art
microservice deployment technologies like Docker~\cite{merkel2014docker} and
Kubernetes~\cite{Hightower}. In particular, the weak and strong dependencies
have been inspired by
Docker Compose~\cite{docker_compose} 
(a language for defining multi-container Docker 
applications) where it is possible to specify
different relationships among microservices 
using, e.g., the {\sf depends\_on} (or {\sf external\_links})
modalities that force (or do not force) a specific startup
order similarly to our strong (or weak) dependencies.
Weak dependencies are also useful to model horizontal scaling, e.g., a load
balancer that is bound to and unbound from many microservice instances during
its life cycle.

In our formalization we also consider resource/cost-aware deployments, taking
inspiration from the {\sf memory} and {\sf CPU} resources found in Kubernetes.
We enrich our model of microservices with the specification of the amount of
resources they need to run. In a deployment, a system of microservices runs
within a set of computation \emph{nodes}. In our model, nodes represent
computational units (e.g., virtual machines in an Infrastructure-as-a-Service
Cloud deployment). Each node has a cost and a set of resources available to the microservices it hosts.

On the model above, we define the \emph{optimal deployment problem} as follows: given
an initial microservice system, a set of available nodes, and a new target
microservice to be deployed, find a sequence of reconfiguration actions that,
once applied to the initial system, leads to a new deployment that includes the
target microservice. The \emph{optimal deployment} has two properties: (a) each used node has at
least as many resources as those needed by the hosted microservices; (b) the
total cost (i.e., the sum of the costs) of the used nodes is minimal. We show
that the optimal deployment problem for microservices is decidable by
presenting an algorithm that
works in three main phases:
(1) generate a set of constraints whose solution 
indicates the microservices to be deployed and their
distribution over available nodes;
(2) generate another set of constraints whose solution
indicates the connections to be established;
(3) synthesize the corresponding deployment plan.
The generated set of constraints are enriched with optimization metrics that minimize the overall 
cost of the computed deployment. 

The algorithm has NEXPTIME complexity because,
in the worst-case, the length of the deployment plan could be
exponential in the size of the input. However, in practice it is 
reasonable to assume that each node can host at most a polynomial 
amount of microservices, as a consequence of its resource limitations.
In this case, the deployment problem is NP-complete and the problem of deploying a system minimizing its total cost is 
an NP-optimization problem. Moreover, having reduced the deployment
problem in terms of constraints, we can exploit state of the art
constraint solvers~\cite{gecode,chuffed,ortools} that are frequently
used in practice to cope with NP-hard problems.

To concretely evaluate our approach,
we consider a real-world microservice architecture,
inspired by the reference email processing pipeline  
from Iron.io~\cite{ironIO}.
We model that architecture in the Abstract
Behavioral Specification (ABS) language, a high-level
object-oriented language that supports deployment modeling~\cite{abs}.
We use our technique to compute two types of deployments: an initial one, with
one instance for each microservice, and a set of deployments to horizontally
scale the system depending on small, medium or large increments in the number
of emails to be processed. The experimental results 
are encouraging in that we were able to compute
deployment plans that add more than 30 new microservice instances, assuming
availability of hundreds of machines of three different types, and guaranteeing
optimality.

\section{The microservice optimal deployment problem}
\label{sec:problem}

\begin{figure}[tb]
\centering
 \includegraphics[width=.8\textwidth]{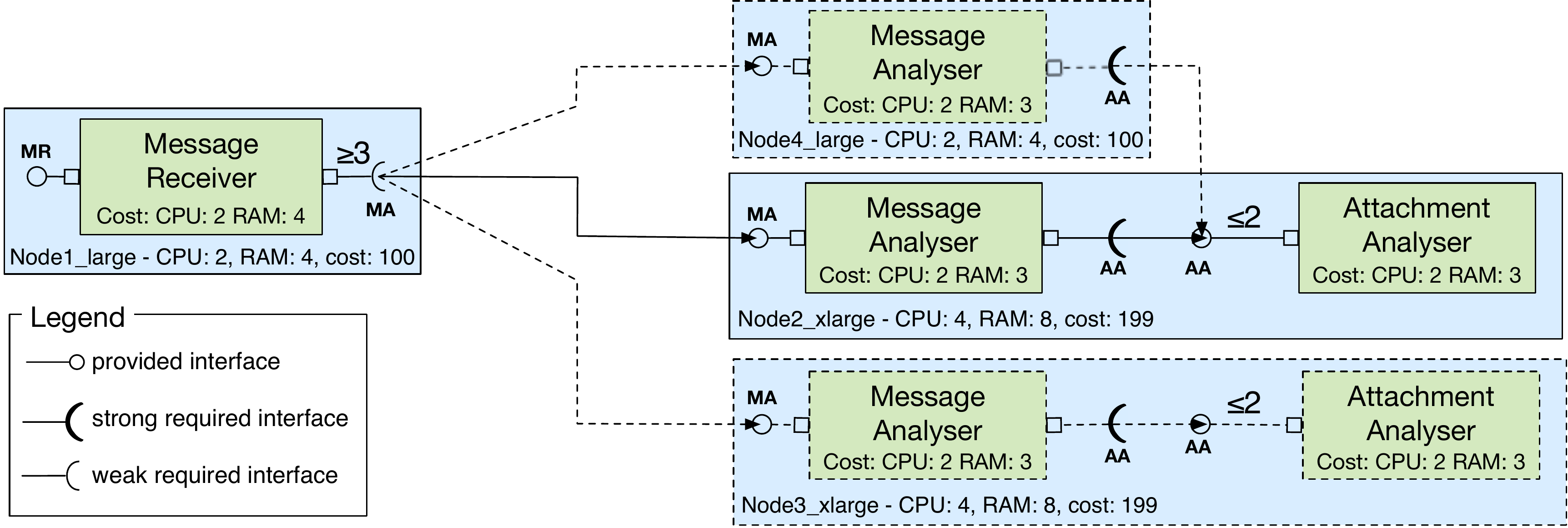}
\caption{Example of microservice deployment (continuous lines: the initial configuration; dashed lines: full configuration).}
\label{fig:informal}
\end{figure}

We model microservice systems as aggregations of components with ports
exposing provided and required interfaces describing offered and required
functionalities, respectively. Microservices are connected by means of bindings 
indicating which port provides the functionality required by another port.
We consider two kinds of requirements: strong required interfaces, that need to
be already fulfilled when the microservice is created, and weak required 
interfaces, that 
must be fulfilled at the end of a deployment (or reconfiguration) plan. 
Microservices are enriched with the specification of the 
resources they need to properly run; such resources are provided
to the microservices by nodes. Nodes can be seen as the unit of computation
executing the tasks associated to each microservice.

As an example, in Fig.~\ref{fig:informal} we have reported
the representation of the deployment of a microservice
system inspired by the email processing pipeline 
that we will discuss in Section~\ref{sec:case_study}.
We consider a simplified pipeline. A {\sf Message Receiver} microservice
handles inbound requests, passing them to a {\sf Message Analyzer} that checks
the email content and sends the attachments for inspection to an {\sf
Attachment Analyzer}. The {\sf Message Receiver} has a port with a \emph{weak}
required interface that can be fulfilled by {\sf Message Analyzer} instances.
The {\sf Message Analyzer} has instead a port with a \emph{strong} required
interface that can be fulfilled by {\sf Attachment Analyzer} instances. In the
second case, the binding between the {\sf Message Analyzer} and the
corresponding {\sf Attachment Analyzer} must be established already when the
{\sf Message Analyzer} is created. In the first case, the bindings between
{\sf Message Receiver} and {\sf Message Analyzer} microservices can be
established afterwards.

The possibility to add new bindings is considered 
in a reconfiguration that,
starting from the initial deployment depicted
in Fig.~\ref{fig:informal} with 
continuous lines, adds the elements
depicted with dashed lines.
In such a reconfiguration,
a 
couple of new instances of {\sf Message Analyzer}
are deployed.
This is done in order to satisfy numerical constraints  
associated to both required and provided interfaces.
For required interfaces, the numerical constraints indicate lower bounds
to the outgoing bindings, while for provided interfaces they specify upper 
bounds to the incoming connections.
In our example, the constraint
$\geq 3$ is associated to the weak required interface of 
{\sf Message Receiver}. In order to fulfill such
a constraint, at least two new instances of {\sf Message Analyzer}
must be added. On the other hand, the constraint $\leq 2$
associated to the interface provided by the {\sf Attachment Analyzer}
implies the creation of a new instance of 
such microservice, in that the initial one cannot serve all 
the three {\sf Message Analyzer}s in the final configuration.

We also model resources: each 
microservice has associated resources
that it consumes (see the {\sf CPU} and {\sf RAM} quantities
associated to the microservices in Fig.~\ref{fig:informal}). 
Resources are provided by nodes, that we represent as containers
for the microservice instances, providing them 
the resources they require. Notice that nodes have also
costs: the total cost of a deployment is the sum of the costs
of the used nodes (e.g., in the example the total cost is 
598 cents per hour, corresponding to the cost of 4 nodes: 2 C4 large and 2 C4 
xlarge virtual machine instances of the Amazon public Cloud).

We now move to the formal definitions.
We assume the following disjoint sets: \Ifaces{} for interfaces, \Resources{} for microservices, and a finite set \Res{} 
for kinds of resources. We use ${\mathbb N}$ to denote natural numbers,
${\mathbb N}^+$ for ${\mathbb N}\setminus\{0\}$,
 and
${\mathbb N}^+_{\infty}$ for ${\mathbb N}^+\cup \{\infty\}$.

\begin{definition}[Microservice type]
  \label{basictypes}
  The set \ResTypesF{} of {\em microservice types},   ranged over by $\ResTy{T}_1,\ResTy{T}_2,\ldots$, contains 5-ples
  $\langle P, D_s, D_w, C, R\rangle$ where:
  \begin{itemize}
  \item $P = (\Ifaces{} \pfun {\mathbb N}_{\infty}^+)$ are the provided interfaces, 
  defined as a partial function from interfaces to corresponding numerical 
  constraints (indicating the maximum number of connected microservices);
  \item $D_s = (\Ifaces{} \pfun {\mathbb N}^+)$ are the \emph{strong} required interfaces, 
  defined as a partial function from interfaces to corresponding numerical 
  constraints (indicating the minimum number of connected microservices);
  \item $D_w = (\Ifaces{} \pfun {\mathbb N})$ are the \emph{weak} required interfaces  (defined as the strong ones, with the difference that also the constraint 0 can be used indicating that it is not strictly necessary to fulfill a weak interface); 
  \item $C \subseteq \Ifaces{}$ are the conflicting interfaces; 
  \item $R = (\Res{} \fun {\mathbb N})$ specifies resource consumption, 
  defined as a total function from resources to corresponding quantities
  indicating the amount of required resources.
  \end{itemize}
  We assume sets $\dom(D_s)$, $\dom(D_w)$ and $C$ to be pairwise disjoint.\footnote{Given a partial function $f$, we use $\dom(f)$ to denote the domain of $f$, i.e.\ the set $\{ e \mid \exists e' : (e,e') \in f \}$.}

\end{definition}
\emph{Notation}: given a microservice type $\ResTy{T} = \langle P, D_s, D_w, C, 
R\rangle$, 
we use the following postfix projections $\GetProv{}$, $\GetSReq{}$, 
$\GetWReq{}$, $\GetConf{}$ and $\GetResource{}$ to decompose it; e.g., 
$\ResTy{T}\GetWReq{}$ returns the partial 
function associating arities to weak required interfacess.
In our example, for instance, the {\sf Message Receiver}
microservice type is such that {\sf Message Receiver}$\GetWReq{({\sf MA})} = 3$
and {\sf Message Receiver}$\GetResource{({\sf RAM})} = 4$.
When the numerical constraints are not explicitly indicated, 
we assume as default value $\infty$ for provided interfaces (i.e., they can satisfy 
an unlimited amount of ports requiring the same interface) 
and 1 for required interfaces (i.e., one connection with a port providing the
same interface is sufficient).

Notice that in the formal definition we consider also conflicting interfaces:
these 
can be used to express conflicts among microservice types that cannot
be both present in a deployment, or cases in which 
a microservice type can have at most one instance (because each additional
instance conflicts with the first one).

We now formalize a well-formedness condition on  
microservice types by requiring that there could be no possible 
cycles of dependencies involving only strong required interfaces. 
Indeed, as strong required interfaces must be already fulfilled 
at the time microservices are instantiated, it is impossible to
deploy mutually strong dependent microservices.

\begin{definition}[Well-formed Universe]
Given a finite set of microservice types $U$ (that we also call \emph{universe}), 
we define the strong dependency graph of $U$ as follows: $G(U)=(U,V)$ with 
$V = \{({\ResTy T},{\ResTy T}') | {\ResTy T}, {\ResTy T}' \in U\ . \ 
\exists p \in \Ifaces \ . \ p \in \dom({\ResTy T}\GetSReq{}) \cap 
\dom({\ResTy T}'\GetProv{}) \}$.
The universe $U$ is well-formed if its strong dependency graph $G(U)$
is acyclic.
\end{definition}
In the following, we always assume universes to be well-formed.
It is worth noting that this does not imply the impossibility
to deploy microservice system with circular dependencies. This remains 
possible, but it is necessary that at least
one weak required interface is involved in the cycle.

\begin{definition}[Nodes]
  \label{nodes}
  The set $\Nodes$ of {\em nodes}   is ranged over by $o_1,o_2,\ldots$ 
  We assume the following information to be associated to each node $o$ in $\Nodes$.
\begin{itemize}
\item A function $R = (\Res{} \fun {\mathbb N})$ that specifies node {\em resource availability}: we use $o\GetResource{}$ to denote such a function.
\item A value in ${\mathbb N}$ that specifies node {\em cost}: we use $o\GetCost{}$ to denote such a value.
\end{itemize}

\end{definition}
As example, in Fig. \ref{fig:informal}, the node {\sf Node1\_large}
is such that {\sf Node1\_large}$\GetResource{({\sf RAM})} = 4$
and {\sf Node1\_large}$\GetCost{} = 100$.

We now define configurations that describe systems composed of microservice
instances and bindings that interconnect them.  A configuration, ranged over by
${\Conf C}_1, {\Conf C}_2,\ldots$, is given by a set of microservice types, a set
of deployed microservices (with their associated type), and a set of
bindings. Formally:
\begin{definition}[Configuration]\label{configurations}
  A \emph{configuration} $\Conf C$ is a 4-ple $\langle Z, T, N, B\rangle$
  where:
  \begin{itemize}
  \item $Z \subseteq \Resources$ is the set of the currently deployed
  {\em microservices};
  \item $T = (Z \fun \ResTy{T})$ are the {\em microservice types}, 
  defined as a function from deployed microservices to microservice types;
  \item $N = (Z \fun \Nodes)$ are the {\em microservice nodes}, 
  defined as a function from deployed microservices to nodes that host them;
  \item $B \subseteq \Ifaces \times Z \times Z$ is the set of {\em bindings},
    namely 3-ples composed of an interface, the microservice that requires that
    interface, and the microservice that provides it; we assume that, for $(p,z_1,z_2) \in B$, the two
    microservices $z_1$ and $z_2$ are distinct and $p \in (\dom(T(z_1)\GetSReq{}) \cup \dom(T(z_1)\GetWReq{})) \cap \dom(T(z_2)\GetProv{})$.
  \end{itemize}
\end{definition}

In our example, if we use {\sf mr} to refer 
to the instance of {\sf Message Receiver}, and {\sf ma} for the 
initially available {\sf Message Analyzer}, we will have the
binding ({\sf MA},{\sf mr},{\sf ma}). Moreover, concerning the microservice
placement function $N$, we have $N({\sf mr})= {\sf Node1\_large}$ and  
$N({\sf ma})= {\sf Node2\_xlarge}$.

We are now ready to formalize the notion of correctness of configuration. We
first define a \emph{provisional correctness}, considering only constraints on strong
required and provided interfaces. Then, we define a general notion of configuration
correctness, considering all kinds of requirements. Intuitively, a
configuration is provisionally correct if, considering its microservice bindings, the numerical
constraints on both strong required and provided interfaces are satisfied. Similarly, a
configuration is correct if it also satisfies the numerical
constraints on weak required interfaces and conflicts are not violated.

\begin{definition}[Provisionally correct configuration]\label{provcorrectconf}
A configuration ${\Conf C}  \!=\! \langle Z, T, N, B\rangle$ is {\em provisionally correct} if, for
  each node $o \!\in\! \ran(N)$, it holds\footnote{Given a (partial) function $f$, we use $\ran(f)$ to denote the range of $f$, i.e.\ the function image set $\{ f(e) \mid e \in \dom(f) \}$. }
$$\forall \, r \!\in\! \Res{} \ldotp \;\, o\GetResource{}(r) \geq \sum_{z \in Z, N(z)=o} T(z)\GetResource{}(r)$$
 and, for each microservice $z \in
  Z$, both following conditions hold: 
\begin{itemize}
\item $(p \mapsto n) \!\in\! T(z)\GetSReq{}$ implies that there exist $n$ distinct microservices
  $z_1,\ldots, z_n$ $\in \!Z \!\setminus\! \{z\}$ such that, for every $1 \leq i \leq n$,
  we have $\langle p, z, z_i\rangle \in B$ and $p \in$ $\dom(T(z_i)\GetProv{})$;
\item $(p \mapsto n) \!\in\! T(z)\GetProv{}$ implies that there exist no $m$
  distinct microservices $z_1,\ldots, z_m \in Z \!\setminus\! \{z\}$, with $m > n$,
  such that, for every $1 \leq i \leq m$, we have $\langle p, z_i, z\rangle
  \in B$ and $p \in \dom(T(z_i)\GetSReq{}) \cup \dom(T(z_i)\GetWReq{})$.
\end{itemize}
\end{definition}

\begin{definition}[Correct configuration]\label{correctconf}
A configuration ${\Conf C}  \!=\! \langle Z, T, N, B\rangle$ is {\em correct} if ${\Conf C}$ is provisionally correct and, for each microservice $z \in
  Z$, both following conditions hold:
\begin{itemize}
\item $(p \mapsto n) \!\in\! T(z)\GetWReq{}$ implies that there exist $n$ distinct microservices
  $z_1,\ldots, z_n$ $\in \!Z \!\setminus\! \{z\}$ such that, for every $1 \leq i \leq n$,
  we have $\langle p, z, z_i\rangle \in B$ and $p \in$ $\dom(T(z_i)\GetProv{})$;
\item $p \!\in\! T(z)\GetConf{}$  implies that, for each $z' \in
  Z\!\setminus\!\{z\}$, we have $p \notin \dom(T(z')\GetProv{})$.
\end{itemize}
\end{definition}

Notice that, in the example in Fig.~\ref{fig:informal}, the initial
configuration (in continuous lines) is only provisionally correct in that the
weak required interface {\sf MA} (with arity 3) of the {\sf Message Receiver}
is not satisfied (because there is only one outgoing binding). The full
configuration --- including also the elements in dotted lines --- is
instead correct: all the constraints associated to the interfaces are
satisfied.

We now formalize how configurations evolve by means
of atomic actions.

\begin{definition}[Actions]\label{actions}
  The set \Actions{} contains the following actions:
  \begin{itemize}
  \item $\mathit{bind}(p,z_1,z_2)$ where $z_1,z_2 \!\in\! \Resources$,  with $z_1 \!\neq\! z_2$, and $p \!\in\!
    \Ifaces$: add a binding between $z_1$ and $z_2$ on port $p$ (which is supposed to be a weak-require port of $z_1$ and a provide port of $z_2$);
  \item $\mathit{unbind}(p,z_1,z_2)$ where $z_1,z_2 \!\in\! \Resources$,  with $z_1 \!\neq\! z_2$, and $p \!\in\!
    \Ifaces$: remove the specified binding on $p$ (which is supposed to be a weak required interface of $z_1$ and a provide port of $z_2$);
  \item $\mathit{new}(z,{\ResTy T}, o, B_s)$ where $z \!\in\! \Resources$, $
  {\ResTy T} \!\in\! \Gamma$, $o \!\in\! \Nodes$ and $B_s \!=\! (\dom({\ResTy
  T}\GetSReq{}) \!\fun\! 2^{\Resources-\{z\}})$; with $B_s$ (representing
  bindings from strong required interfaces
  in {\ResTy T} to sets of microservices) being such that, for each $p \in \dom({\ResTy T}\GetSReq{})$, it holds $|B_s(p)|\geq {\ResTy T}\GetSReq{}(p)$: add a new microservice $z$ of type $\ResTy T$ hosted in $o$ and bind each of its strong required interfaces to a set of microservices as described by $B_s$;\footnote{Given sets $S$ and $S'$ we use: $2^S$ to denote the power set of $S$, i.e. the set $\{S' \mid S \subseteq S \}$; $S-S'$ to denote set difference; and $|S|$ to denote the cardinality of $S$.}
  \item $\mathit{del}(z)$ where $z \!\in\! \Resources$: remove the microservice $z$
    from the configuration and all bindings involving it.
  \end{itemize}
\end{definition}

In our example, assuming that the
initially available {\sf Attachment Analyzer}
is named {\sf aa},
we have that the action to create the initial instance
of {\sf Message Analyzer} is 
$\mathit{new}({\sf ma},{\sf Message Analyzer}, 
{\sf Node2\_xlarge}, ({\sf AA} \mapsto \{{\sf aa}\}))$.
Notice that it is necessary to establish the binding with
the {\sf Attachment Analyzer} because of the corresponding 
strong required interface.

The execution of actions can now be formalized using a labeled transition
system on configurations, which uses actions as labels.

 \begin{definition}[Reconfigurations]\label{reconfigurations}
   Reconfigurations are denoted by transitions $\Conf{C} \trans{\alpha}
   \Conf{C}'$ meaning that the execution of $\alpha \in \Actions$ on the
   configuration \Conf{C} produces a new configuration $\Conf{C}'$. The
   transitions from a configuration $\Conf{C}=\langle Z,T,N,B \rangle$ are
   defined as follows:
\footnotesize
 $$
 \begin{array}{@{}lcr@{}}  
     \begin{array}[t]{l}
       \Conf{C} \trans{\mathit{bind}(p,z_1,z_2)} 
       \langle Z,T,N,B \cup \langle p,z_1,z_2\rangle \rangle \\
       \quad
       \begin{array}[t]{l}
         \mbox{if $\langle p,z_1,z_2\rangle \not\in B$  and} \\
         \mbox{$p \in \dom(T(z_1)\GetWReq) \cap \dom(T(z_2)\GetProv) $}
       \end{array}
     \end{array} & &
   \begin{array}[t]{l}
     \Conf{C} \trans{\mathit{unbind}(p,z_1,z_2)} 
     \langle Z,T,N,B \!\setminus\! \langle p,z_1,z_2\rangle \rangle \\
     \quad       
     \begin{array}[t]{l}
     \mbox{if $\langle p,z_1,z_2\rangle \in B$ and} \\
         \mbox{$p \in \dom(T(z_1)\GetWReq) \cap \dom(T(z_2)\GetProv)  $}
     \end{array}
   \end{array}
   \\ \\
     \begin{array}[t]{l}
       \Conf{C} 
       \trans{\mathit{new}(z,{\CompTy T},n,B_s)} 
       \langle Z \cup \{z\}, T', N', B' \rangle \\
       \quad
       \begin{array}{l}
         \mbox{if $z \not\in Z$ and} \\
	   \mbox{$\forall \, p \in \dom({\ResTy T}\GetSReq{}) \ldotp \, \forall z' \in B_s(p) \ldotp$} \\
         \mbox{$\hspace{.5cm} p \in \dom(T(z')\GetProv{})$ and} \\
         \mbox{$T' = T \cup \{ (z \mapsto {\CompTy T}) \}$ and} \\
         \mbox{$N' = N \cup \{ (z \mapsto o) \}$ and} \\
         \mbox{$B' = B \cup \{ \langle p,z,z'\rangle \mid z' \in B_s(p) \}$} 
       \end{array}
     \end{array}
     &
     &
     \begin{array}[t]{l}
       \Conf{C} \trans{\mathit{del}(z)} 
       \langle Z \!\setminus\! \{z\}, T', N' , B' \rangle \\
       \quad
       \begin{array}{l}
         \mbox{if $T'=\{(z' \mapsto {\CompTy T}) \in T \mid z \neq z'\}$ and}\\
         \mbox{$N'=\{(z' \mapsto o) \in N \mid z \neq z'\}$ and}\\
         \mbox{$B'=\{\langle p,z_1,z_2\rangle \in B \mid z \not\in \{z_1,z_2\}\}$}
       \end{array}
     \end{array}
 \end{array}
 $$
\end{definition}

A \emph{deployment plan} is simply a sequence of actions that
transform a provisionally correct configuration (without violating provisional correctness along the way) and, finally, reach a correct configuration.

\begin{definition}[Deployment plan] \label{deployment_plan}
   A {\em deployment plan} $\mathsf{P}$ from a provisionally correct configuration ${\mathcal C}_{0}$ is a 
sequence of actions $\alpha_{1}, \dots, \alpha_{m}$ such that:
\begin{itemize}
\item
there exist ${\mathcal C}_{1}, \dots, 
{\mathcal C}_{m}$ provisionally correct configurations, with ${\mathcal C}_{i-1} \trans{\alpha_{i}}  
{\mathcal C}_{i}$ for $1 \leq i \leq m$, and
\item
${\mathcal C}_{m}$ is a correct configuration.
\end{itemize}
Deployment plans are also denoted with ${\mathcal C}_{0}
   \trans{\alpha_1} {\mathcal C}_{1} \trans{\alpha_{2}} \cdots
   \trans{\alpha_{m}} {\mathcal C}_{m}$.
\end{definition}

In our example, a deployment plan that reconfigures the
initial provisionally correct configuration into the 
final correct one is as follows:
a $\mathit{new}$ action to create the new instance of 
{\sf Attachment Analyzer}, followed by two $\mathit{new}$ actions
for the new {\sf Message Analyzer}s (as commented
above, the connection with the {\sf Attachment Analyzer}
is part of these $\mathit{new}$ actions), and finally
two $\mathit{bind}$ actions to
connect the {\sf Message Receiver} to the two new instances of
{\sf Message Analyzer}.

We now have all the ingredients to define the \emph{optimal deployment
problem}, that is our main concern: given a universe of microservice types, a
set of available nodes and an initial configuration, we want to know whether
and how it is possible to deploy at least one microservice of a given
microservice type $\ResTy{T}$ by optimizing the overall cost of nodes hosting
the deployed microservices.

\begin{definition}[Optimal deployment problem] \label{problem}
  The \emph {optimal deployment problem} has, as input, a finite well-formed universe $U$ of microservice
  types, a finite set of available nodes $O$, an initial provisionally correct configuration ${\mathcal C}_{0}$ and a microservice type $\ResTy{T}_t \in U$.  
  The output is:
\begin{itemize}
\item A {\bf deployment plan} 
  $\Plan{P} = {\mathcal C}_{0} \trans{\alpha_1}
   {\mathcal C}_{1} \trans{\alpha_{2}} \cdots \trans{\alpha_{m}} 
   {\mathcal C}_{m}$ such that
\begin{itemize}
\item for all $\Conf{C}_{i}=\langle Z_i,T_i,N_i,B_i \rangle$, with $1 \leq i \leq m$, it holds $\forall z \in Z_i \ldotp \, T_i(z) \in U \wedge N_i(z) \in O$, and
\item $\Conf{C}_{m}=\langle Z_m,T_m,N_m,B_m \rangle$ satisfies $\exists z \in Z_m : T_i(z) = \ResTy{T}_t$;
\end{itemize}
if there exists one. In particular, among all deployment plans satisfying the constraints above, one that minimizes $\sum_{z \in Z_m} N_m(z)\GetCost$ (i.e. the overall cost of nodes in the last configuration ${\mathcal C}_{m}$), is outputed.
\item {\bf no}    (stating that no such plan exists); otherwise.
\end{itemize}
\end{definition}

We are finally ready to state our main result on the decidability of the
optimal deployment problem. To prove the result we describe an approach that
splits the problem in three incremental phases: (1) the first phase checks if
there is a possible solution and assigns microservices to deployment nodes,
(2) the intermediate phase computes how the microservices need to be connected
to each other, and (3) the final phase synthesizes the corresponding
deployment plan.

\begin{theorem}
\label{th:decidability}
The optimal 
deployment problem is decidable.
\end{theorem}

\begin{proof}
The proof is in the form of an algorithm that solves the optimal deployment
problem. We assume that the input to the problem to be solved is given by $U$
(the microservice types), $O$ (the set of available nodes), ${\mathcal C}_{0}$
(the initial provisionally correct configuration), and $\ResTy{T}_t \in U$
(the target microservice type). We use $\Ifaces(U)$ to denote the set of
interfaces used in the considered microservice types, namely
$\Ifaces(U)=\bigcup_{\ResTy{T} \in U} \dom(\ResTy{T}\GetSReq{}) \cup
\dom(\ResTy{T}\GetWReq{}) \cup \dom(\ResTy{T}\GetProv{}) \cup
\ResTy{T}\GetConf{}$.

\noindent
The algorithm is based on three phases.

\vspace{5pt}
\emph{Phase 1}
The first phase consists of the generation of a set of constraints
that, once solved, indicates how many instances should be created
for each microservice type $\ResTy{T}$ (denoted with $\texttt{inst}(\ResTy{T})$),
how many of them should be deployed on node $o$ 
(denoted with $\texttt{inst}(\ResTy{T},o)$),
and how many bindings should be established for each interface $p$
from instances of type $\ResTy{T}$ --- considering both 
weak and strong required interfaces --- 
and instances of type $\ResTy{T}'$
(denoted with $\texttt{bind}(p,\ResTy{T},\ResTy{T}')$). We also generate
an
optimization function that guarantees that the generated
configuration is minimal w.r.t.\ its total cost.

We now incrementally report the generated constraints.
The first group of constraints deals with the number of bindings:

\begin{subequations}
\label{constraints-eq}
\footnotesize
\begin{align}
& \bigwedge_{p \in \Ifaces(U)} \quad \bigwedge_{\ResTy{T} \in U, 
\; p \in \dom(\ResTy{T}\GetSReq{})}
\ResTy{T}\GetSReq{}(p) \cdot \texttt{inst}(\ResTy{T})
\leq \sum_{\ResTy{T'}\in U}
\texttt{bind}(p,\ResTy{T},\ResTy{T'})
\label{constraint-require}\\
& \bigwedge_{p \in \Ifaces(U)} \quad \bigwedge_{\ResTy{T} \in U, 
\; p \in \dom(\ResTy{T}\GetWReq{})}
\ResTy{T}\GetWReq{}(p) \cdot \texttt{inst}(\ResTy{T})
\leq \sum_{\ResTy{T'}\in U}
\texttt{bind}(p,\ResTy{T},\ResTy{T'})
\label{constraint-require2}\\
& \bigwedge_{p \in \Ifaces(U)} \quad \bigwedge_{\ResTy{T}\in U, \; \ResTy{T}\GetProv{}(p) < \infty}
\ResTy{T}\GetProv{}(p) \cdot \texttt{inst}(\ResTy{T}) \geq 
\sum_{\ResTy{T'}\in U}
\texttt{bind}(p,\ResTy{T'},\ResTy{T})
\label{constraint-provide}\\
& \bigwedge_{p \in \Ifaces(U)} \quad \bigwedge_{\ResTy{T}\in  U , 
\; \ResTy{T}\GetProv{}(p) = \infty}
\texttt{inst}(\ResTy{T}) = 0\ \ \Rightarrow\ \
\sum_{\ResTy{T'}\in U}
\texttt{bind}(p,\ResTy{T'},\ResTy{T})=0
\label{constraint-provide2}\\
& \bigwedge_{p \in \Ifaces(U)} \quad \bigwedge_{\ResTy{T}\in  U , 
\; p \notin \dom(\ResTy{T}\GetProv{})}
\;\;\; \sum_{\ResTy{T'}\in U}
\texttt{bind}(p,\ResTy{T'},\ResTy{T})=0
\label{constraint-provide3}
\end{align} 
\end{subequations}

Constraint~\ref{constraint-require} and \ref{constraint-require2} guarantee 
that there are
enough bindings to satisfy all the required interfaces, considering
both strong and weak requirements.
Symmetrically, constraint~\ref{constraint-provide}
guarantees that the number of bindings is not greater than
the total available capacity, computed as the sum of the single capacities of
each provided interface.
In case the capacity
is unbounded (i.e., $\infty$), it is sufficient to have at least one instance
that activates such port to support any possible requirement (see
constraint~\ref{constraint-provide2}). 
Finally, constraint~\ref{constraint-provide3} guarantees that no
binding is established connected to provided interfaces of
microservice types that are not deployed.

The second group of constraints deals with the number
of instances of microservices to be deployed.

\begin{subequations}
\label{constraints-eq-1}
\footnotesize
\begin{align}
& \texttt{inst}(\ResTy{T}_t) \geq 1
\label{constraint-target}\\
& \bigwedge_{p \in \Ifaces(U)} \quad
\bigwedge_{\substack{
\ResTy{T}\in  U , \\
p \in \ResTy{T}\GetConf}} \quad
\bigwedge_{\substack{
\ResTy{T'}\in  U-\{\ResTy{T}\} , \\
p \in \dom(\ResTy{T'}\GetProv{}) }}
\texttt{inst}(\ResTy{T}) > 0 
\ \ \Rightarrow\ \
\texttt{inst}(\ResTy{T'}) = 0
\label{constraint-conflict1}\\
& \bigwedge_{p \in \Ifaces(U)} \quad
\bigwedge_{\substack{
\ResTy{T}\in  U , \; p \in \ResTy{T}\GetConf \; \wedge \\
p \in \dom(\ResTy{T}\GetProv{}) }}
\texttt{inst}(\ResTy{T}) \leq 1
\label{constraint-conflict2}\\
& \bigwedge_{p \in \Ifaces(U)}
\quad \bigwedge_{\ResTy{T} \in  U}
\quad \bigwedge_{\ResTy{T'}\in  U-\{\ResTy{T}\}}
\texttt{bind}(p,\ResTy{T},\ResTy{T'}) \leq
\texttt{inst}(\ResTy{T}) \cdot
\texttt{inst}(\ResTy{T'})
\label{constraint-distinct}\\
& \bigwedge_{p \in \Ifaces(U)}
\quad \bigwedge_{\ResTy{T} \in  U}
\texttt{bind}(p,\ResTy{T},\ResTy{T}) \leq
\texttt{inst}(\ResTy{T}) \cdot
(\texttt{inst}(\ResTy{T})-1)
\label{constraint-distinct2}
\end{align} 
\end{subequations}

The first constraint~\ref{constraint-target} guarantees the
presence of at least one instance of the target microservice.
Constraint~\ref{constraint-conflict1} guarantees that no two instances
of different types will be created if one activates a conflict on 
an interface provided by the other one.
Constraint~\ref{constraint-conflict2}, consider the other case
in which a type activates the same interface both in 
conflicting and provided modality:
in this case, at most one instance of such type can be created.
Finally, the constraints~\ref{constraint-distinct}
and~\ref{constraint-distinct2} guarantee that there are enough pairs of
distinct instances to establish all the necessary bindings. Two distinct
constraints are used: the first one deals with bindings between microservices of
two different types, the second one with bindings between microservices of
the same type.

The last group of constraints deals with the distribution of microservice
instances over the available nodes $O$.

\begin{subequations}
\label{binpacking-eq}
\footnotesize
\begin{align}
& \texttt{inst}(\ResTy{T}) = \sum_{o \in O} \texttt{inst}(\ResTy{T},o) 
\label{bin-packing1}\\
& \bigwedge_{r \in \Res{}} \bigwedge_{o \in O} 
\sum_{\ResTy{T} \in U} \texttt{inst}(\ResTy{T},o) \cdot \ResTy{T}\GetResource{}(r) \leq 
o\GetResource{}(r) 
\label{bin-packing2}\\
& \bigwedge_{o \in O} \big(\sum_{\ResTy{T} \in U} \texttt{inst}(\ResTy{T},o) > 0 
\big)
\Leftrightarrow \texttt{used}(o) 
\label{bin-packing3}\\
& \min \sum_{o \in O, \, \texttt{used}(o) }  o\GetCost 
\label{bin-packing4}
\end{align}
\end{subequations}

Constraint~\ref{bin-packing1} simply formalizes the relationship among
the variables $\texttt{inst}(\ResTy{T})$ and $\texttt{inst}(\ResTy{T},o)$
(the total amount of all instances of a microservice type, should
correspond to the sum of the instances locally deployed on each node).
Constraint~\ref{bin-packing2} checks that each node has enough resources
to satisfy the requirements of all the hosted microservices.
The last two constraints define the optimization function
used to minimize the total cost: constraint~\ref{bin-packing3}
introduces the boolean variable $\texttt{used}(o)$ which is true
if and only if node $o$ contains at least one microservice instance;
constraint~\ref{bin-packing4} is the function to be minimized, i.e.,
the sum of the costs of the used nodes.

These constraints, and the optimization function, are expected to 
be given in input to a constraint/optimization solver. If a solution
is not found
it is not possibile
to deploy the required microservice system; otherwise, the
next phases of the algorithm are executed to synthesize
the optimal deployment plan.

\vspace{5pt}
\emph{Phase 2}
The second phase consists of the generation of another set of constraints
that, once solved, indicates the bindings to be established between
any pair of microservices to be deployed. More precisely,
for each type $\ResTy{T}$ such that $\texttt{inst}(\ResTy{T})>0$,
we use $s_i^\ResTy{T}$, with $1 \leq i \leq \texttt{inst}(\ResTy{T})$,
to identify the microservices of type $\ResTy{T}$ to be deployed.
We also assume a function $N$ that associates 
microservices to available nodes $O$, which is compliant
with the values $\texttt{inst}(\ResTy{T},o)$ already computed in Phase 1,
i.e., given a type $\ResTy{T}$
and a node $o$, the number of $s_i^\ResTy{T}$, with $1 \leq i \leq \texttt{inst}(\ResTy{T})$, such that $N(s_i^\ResTy{T}) = o$ coincides with 
$\texttt{inst}(\ResTy{T},o)$.

In the constraints below we use the variables 
$\texttt{b}(p,s_i^\ResTy{T},s_j^{\ResTy{T}'})$
(with $i \neq j$):
the value of such variables is $1$ if there is a connection between 
the required interface $p$ of $s_i^\ResTy{T}$ and 
the corresponding provided interface of $s_j^{\ResTy{T}'}$, $0$ otherwise.
We also make use of an auxiliary total function 
$\it{limProv}(\ResTy{T},p)$ that extends $\ResTy{T}\GetProv{}$
associating 0 to the interfaces outside its domain.

\begin{subequations}
\label{binding-eq}
\footnotesize
\begin{align}
&
\bigwedge_{\ResTy{T} \in  U}
\bigwedge_{p \in \Ifaces(U)}
\bigwedge_{i \in 1 \ldots \texttt{inst}(\ResTy{T})}
\sum_{j \in (1 \ldots \texttt{inst}(\ResTy{T}))\setminus \{i\}} \texttt{b}(p,s_i^\ResTy{T},s_j^\ResTy{T}) \leq \it{limProv}(\ResTy{T},p)
\label{const_binding1}\\
& 
\bigwedge_{\ResTy{T} \in  U}
\bigwedge_{p \in \dom(\ResTy{T}\GetSReq{})}
\bigwedge_{i \in 1 \ldots \texttt{inst}(\ResTy{T})}
\sum_{j \in (1 \ldots \texttt{inst}(\ResTy{T}))\setminus \{i\}} 
\texttt{b}(p,s_i^\ResTy{T},s_j^\ResTy{T}) \geq \ResTy{T}\GetSReq{(p)}
\label{const_binding2}\\
& 
\bigwedge_{\ResTy{T} \in  U}
\bigwedge_{p \in \dom(\ResTy{T}\GetWReq{})}
\bigwedge_{i \in 1 \ldots \texttt{inst}(\ResTy{T})}
\sum_{j \in (1 \ldots \texttt{inst}(\ResTy{T}))\setminus \{i\}} 
\texttt{b}(p,s_i^\ResTy{T},s_j^\ResTy{T}) \geq \ResTy{T}\GetWReq{(p)}
\label{const_binding3}\\
& 
\bigwedge_{\ResTy{T} \in  U}
\bigwedge_{p \notin \dom(\ResTy{T}\GetSReq{}) \cup \dom(\ResTy{T}\GetWReq{})}
\bigwedge_{i \in 1 \ldots \texttt{inst}(\ResTy{T})}
\sum_{j \in (1 \ldots \texttt{inst}(\ResTy{T}))\setminus \{i\}} 
\texttt{b}(p,s_i^\ResTy{T},s_j^\ResTy{T}) = 0
\label{const_binding4}
\end{align}
\end{subequations}

Constraint~\ref{const_binding1} considers the provided interface capacities
to fix upper bounds to the bindings to be established,
while contraints~\ref{const_binding2} and~\ref{const_binding3}
fix lower bounds based on the required interface capacities,
considering both the weak (see~\ref{const_binding2})
and the strong (see~\ref{const_binding3}) ones.
Finally, constraint~\ref{const_binding4} indicates that it is not
possible to establish connections on interfaces that are not
required.

A solution for these 
constraints 
exists because the 
constraints~\ref{constraint-require}~$\dots$~\ref{constraint-distinct2}
(already solved during Phase 1)
guarantee that the configuration to be synthesized contains
enough capacity on the provided interfaces to satisfy
all the required interfaces.

\vspace{5pt}
\emph{Phase 3}
In this last phase we synthesize the deployment plan that, when applied to the 
initial configuration ${\mathcal C}_{0}$, reaches a new configuration 
${\mathcal C}_{t}$ with nodes, microservices and bindings as computed in 
the first two phases of the algorithm. Without loss of generality, in this
proof we show the existence of a simple plan that first removes 
the elements in the initial configuration and then deploys 
the target configuration from scratch.
However, as also discussed in Section~\ref{sec:case_study},
in
practice
it is possible to define elaborated planning mechanisms that
re-use microservices already deployed.

Reaching an empty 
configuration is a trivial task since it is always possible to perform in the initial configuration
unbind actions for all the bindings connected to  
weak required interfaces. Then the microservices can be deleted since for 
the well-formedness of the system it is possible to order, using a topological 
sort, the microservices to be removed without violating any strong required
interface (e.g., first remove the microservice not requiring anything and 
repeat until all the microservices have been deleted).

The deployment of the target configuration follows a similar pattern. Given the
distribution of microservices over nodes
(computed in the first phase) and the corresponding bindings
(computed in the second phase), the microservices can be 
created by following a topological sort considering the microservices
dependencies following from the strong required interfaces.
When all the microservices are deployed on the corresponding nodes, the 
remaining bindings (on weak required ports) may be added in any possible order.
\qed
\end{proof}

\begin{remark}
\label{rem:bindOpt}
The constraints generated during Phase 2 of the algorithm, in order
to establish the microservice bindings,
are expected to be given in input
to a constraint/optimization solver. One can enrich such constraints with metrics to optimize, e.g., 
the number of local bindings (i.e., give a preference to the connections
among microservices hosted in the same node):
\[
\min \sum_{(\ResTy{T},\ResTy{T}' \in U), i \in 1 \ldots \texttt{inst}(\ResTy{T}),
j \in 1 \ldots \texttt{inst}(\ResTy{T}'), p \in \Ifaces(U)\ .\ 
N(s_i^\ResTy{T}) \neq N(s_j^{\ResTy{T}'})} \texttt{b}(p,{s_i^\ResTy{T}},s_j^{\ResTy{T}'})
\]
Another example, used in the case study discussed
in Section~\ref{sec:case_study}\footnote{We modeled a load balancer as 
a microservice having a weak required interface 
with arity 0 that can be provided by its back-end service. By adopting the 
above metric, the synthesized 
configuration connected all possible services to such required interface 
in order to allow the load balancer to forward requests to all of them.
}, is the following metric that maximizes 
the number of bindings:
\[
\max \sum_{s_i^\ResTy{T}, s_j^{\ResTy{T}'}, p \in \Ifaces(U)} \texttt{b}(p,s_i^\ResTy{T}, s_j^{\ResTy{T}'})
\]
\end{remark}

From the complexity point of view, it is possible to show that the decision 
versions of the optimization 
problem solved in Phase 1 is NP-complete, in Phase 2 is in NP, 
while the planning in Phase 3 is synthesized in 
polynomial time. Unfortunately, due to the 
fact that numeric constraints can be represented in 
log space, the output of Phase 2 requiring the enumeration of all the 
microservices to deploy can be exponential in the size of the output of 
Phase 1 (indicating only the total number of instances for each type). 
For this reason, the optimal deployment problem is in NEXPTIME. 
However,
in practice, due to the resource usage of the microservices, the number of 
microservices to be deployed 
can be assumed to be polynomial in the size of the input. In this case the optimal 
deployment problem becomes an NP-optimization problem and its 
decision 
version is NP-complete.
A formal proof of the 
complexity of the problem is available in Appendix \ref{app:complexity}.

\section{Application of the technique to the case-study}
\label{sec:case_study}

Given the asymptotic complexity of our solution 
(NP under the assumption of polynomial size of the target configuration) 
we have decided to evaluate its applicability  in practice
by considering
a real-world microservice architecture,
namely the email processing pipeline described in~\cite{ironIO}.
\begin{figure}[t]
  \centering
  \includegraphics[width=.8\textwidth]{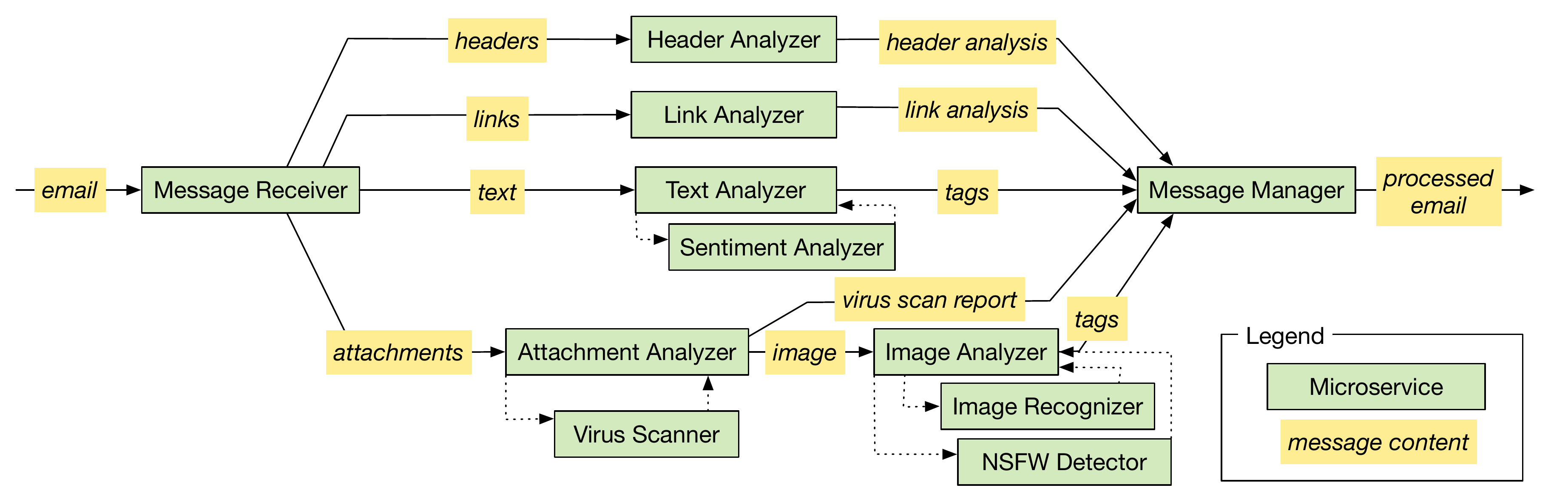}
  \caption{Microservice architecture for email processing.}
  \label{fig:arch}
\end{figure}
The considered architecture
separates and routes the components found in an email (headers, links, text,
attachments) into distinct, parallel sub-pipelines with specific tasks (e.g., remove malicious attachments, tag the content of the mail). We report in
Fig.~\ref{fig:arch} a depiction of the architecture.
From left to right, when an email reaches the
\textsf{Message Receiver}, it sends each component
into a specific sub-pipeline. In the sub-pipelines, some
microservices --- e.g., \textsf{Text Analyzer} and \textsf{Attachment
Analyzer} --- coordinate with other microservices --- e.g.,
\textsf{Sentiment Analyzer} and \textsf{Virus Scanner} --- to process their
inputs.
Each microservice in the architecture has a given resource consumption
(expressed in terms of CPU 
and memory).
As expected, the processing of each
email component entails a specific load.
Some microservices can handle large inputs, e.g., in the range of 40K
simultaneous requests --- like the \textsf{Header Analyzer} that processes
short and uniform inputs. Other microservices sustain heavier computations ---
like the \textsf{Image Recognizer} --- and can handle smaller simultaneous inputs, e.g., in the range of 10K requests.

To model the system above, we use the Abstract Behavioral Specification (ABS)
language, a high-level object-oriented language that supports deployment
modeling~\cite{abs}. ABS is agnostic w.r.t.\ deployment platforms (Amazon AWS,
Microsoft Azure) and technologies (e.g., Docker or Kubernetes) and it offers
high-level deployment primitives for the creation of new \emph{deployment
components} and the instantiation of objects inside them. Here, we use ABS
deployment components as computation nodes, ABS objects as microservice
instances, and ABS object references as bindings. Finally, to describe the requirements in our model, we use ABS with \tool~\cite{smartdepl}, an extension
that supports dependency annotations (e.g., from other classes, available
resources) in ABS classes. We use annotations to model strong
required interfaces as class dependencies and weak required interfaces as
object references, which can be passed to running objects. We define a class
for each microservice type, plus one \emph{load balancer} class for each
microservice type. A load balancer distributes requests over a set of
microservice instances that can scale horizontally. Finally, we model nodes
over three popular Amazon EC2 instances: \texttt{c4\_large},
\texttt{c4\_xlarge}, and \texttt{c4\_2xlarge} (with the corresponding provided
resources and costs).

\begin{center}
{\sffamily \scriptsize
\newcolumntype{P}[1]{>{\centering\arraybackslash}p{#1}}
\begin{tabular}{|P{.4\textwidth}|P{.14\textwidth}|P{.14\textwidth}|P{.14\textwidth}|P{.14\textwidth}|}
    \hline 
    \textbf{Microservice (max computational load)} 
& \textbf{Initial (10K)} & \textbf{+20K} & \textbf{+50K} & \textbf{+80K} \\ 
    \hline 
    \textbf{MessageReceiver}($\infty$) & 1 & +0 & +0 & +0 \\ 
    \hline 
    \textbf{MessageParser}(40K) & 1 & +0 & +1 & +0 \\ 
    \hline 
    \textbf{HeaderAnalyser}(40K) & 1 & +0 & +1 & +0 \\ 
    \hline 
    \textbf{LinkAnalyser}(40K) & 1 & +0 & +1 & +0 \\ 
    \hline 
    \textbf{TextAnalyser}(15K) & 1 & +1 & +2 & +2 \\ 
    \hline 
    \textbf{SentimentAnalyser}(15K) & 1 & +3 & +4 & +6 \\ 
    \hline 
    \textbf{AttachmentsManager}(30K) & 1 & +1 & +2 & +2 \\ 
    \hline 
    \textbf{VirusScanner}(13K) & 1 & +3 & +4 & +6 \\ 
    \hline 
    \textbf{ImageAnalyser}(30K) & 1 & +1 & +2 & +2 \\ 
    \hline 
    \textbf{NSFWDetector}(13K) & 1 & +3 & +4 & +6 \\ 
    \hline 
    \textbf{ImageRecognizer}(13K) & 1 & +3 & +4 & +6 \\ 
    \hline 
    \textbf{MessageAnalyser}(70K) & 1 & +1 & +2 & +2 \\ 
    \hline 
  \end{tabular}}
\end{center}
 In the table above, we report the result of our algorithm
w.r.t.\ four incremental configurations: the initial in column 2 and under
incremental loads in 3--5. We also consider an availability of 40 nodes for
each of the three node types.
In the first column of the Table, next to a microservice type,
we report its corresponding maximum computational load. As visible in columns 2--5, different maximal computational loads imply
different scaling factors w.r.t.\ a given number of simultaneous requests. In the initial configuration
we consider 10K simultaneous requests and we have one instance of each
microservice type (and of the corresponding load balancer). The other
deployment configurations deal with three scenarios of horizontal scaling,
assuming three increasing increments of inbound messages (20K, 50K, and 80K). In the three scaling scenarios, we do not implement the planning algorithm 
described in Phase 3 of the proof of Theorem~\ref{th:decidability}. Contrarily, 
we take advantage of the presence of the load balancers and, as described in 
Remark~\ref{rem:bindOpt}, we achieve a similar result with an optimization 
function that maximizes the number of bindings of the load balancers.
For every scenario, we generated
automatically the ABS code for the 
plan that deploys an optimal configuration, using a time cap of half an hour 
for every deployment scenario.\footnote{Half and hour is a reasonable time cap 
for the computation of the deployment plans at the design phase, as in our 
case. For run time usage, trying to reduce the running times to few 
minutes (i.e., the times it usually takes to start a new virtual machine in a 
public cloud) is left as a future work.}

The ABS code modeling the system and the generated code are publicly available
at~\cite{abs_code_ex}. A graphical representation of the initial configuration
is available in Appendix~\ref{app:configuration}.

\section{Related Work and Conclusion}

With the current popularity of Cloud Computing, the problem of automating
application deployment has attracted a lot of attention and many system management
tools exists~\cite{kanies-puppet-login,mcollective,chef,ansible}. Those tools support the specification of deployment plans but they do not support automatic distribution of software instances over the
available machines. For these
reasons, those tools do not solve the deployment problem as defined in this paper, 
but are just deployment engines to concretely execute deployment plans.

The proposals closest to ours are those by Feinerer~\cite{Feinerer13} and by
Fischer at al.~\cite{engage}. Both proposals rely on a solver to plan
deployments. The first is based on the UML component model, which includes
conflicts and dependencies, but lacks the modeling of nodes. The second does
not support conflicts in the specification language. Neither proposals support
the computation of optimal deployments.

Our work is inspired by the Aeolus component model~\cite{infCom14,concur15},
the Zephyrus configuration optimizer~\cite{zephyrus2}, and
ConfSolve~\cite{ConfSolve}. The Aeolus model paved the way to reason on
deployment and reconfiguration, proving some decidability results. Zephyrus is
a configuration tool grounded on the Aeolus model and underpins the first
phase of our approach. Similarly, ConfSolve relies on constraint solving
techniques to propose an optimal allocation of virtual machines to servers,
and of applications to virtual machines. Both tools ignore the problem of
synthesizing a low-level plan to reach the final configuration which, in the
general case, has been proven undecidable.
In this work, by considering microservices, we prove that the generation of
the plan becomes decidable and thus fully automatable, from the synthesis of
the optimal configuration to the generation of the actions to deploy it. We
show a practical application of our approach on a non-trivial example of
microservice architecture, modeled in the Abstract Behavioral Specification
(ABS) language. As a result, we synthesize an optimal initial configuration
and different scaling scenarios, generating the deployment actions directly in
ABS.

Regarding autoscaling, existing
solutions~\cite{swarm,mesos,cloud_watch,Hightower} support the automatic
increase or decrease of the number of instances of a service/container, when
some conditions (e.g., CPU average load greater than 80\%) are met. Our work is
an example of how we can go beyond single-component horizontal scaling
policies. Contrarily, our approach supports the computation of optimal
horizontal scaling operations involving at the same time more than one service,
thus enabling to reason on autoscaling operation at the application
level.

As a future work we are interested in investigating local search 
approaches to speed up the
solution of the optimization problems involved in the deployment problem. This
will allow us to use our approach at run time when responses are needed in a 
short amount
of time (e.g., minutes)
at the price of losing the optimality guarantee of the solutions. This is 
probably an inevitable trade-off due to the NP-hardness of the 
optimal deployment problem.
 
\bibliographystyle{splncs04}
\bibliography{biblio}

\newpage

\appendix

\section{Optimal Deployment Problem Complexity}
\label{app:complexity}

\begin{theorem}
The optimal 
deployment problem is in NEXPTIME. If the number of microservices to be 
deployed is polynomial in the size of the input, the problem is an an 
NP-optimization (NPO) problem and its decision 
version is NP-complete.
\end{theorem}

\begin{proof}
The proof derives from the fact that the decision version of the optimization 
problem solved in phase 1 is NP-complete, the decision version of the 
optimization problem solved in phase 2 is in NP, and the problem in phase 
3 is polynomial.

Due to the fact that numeric constraints can be represented in 
log space, the input of phase 2 can be exponential in the size of the output of 
phase 1. This for instance happens when the target component requires an 
interface $p$ with numerical constraint $\geq n$ and when all the components 
providing the interface $p$ have numerical constraint equal to 1. The 
solution in phase 1 will require the deployment of $n$ microservices and can 
be represented in $O(\log(n))$ space. However, phase 2 requires the list of 
microservices to be deployed and this is represented only in $O(n)$ space. 

This makes the optimal deployment problem an NEXPTIME problem.
However, when the microservices to be deployed in the final configuration 
are polynomially bounded in the size of the input\footnote{Note that this is a 
reasonable expectation in practice because microservices require resources hence
only a limited number of them are installable on the same node.}, the 
optimal deployment problem becomes an NPO problem due to the fact that its 
decision version is an NP-complete problem, being equivalent to the 
execution in sequence of 2 NP-complete 
problems.

We will now proceed by proving the complexity of the 3 phases used to solve the 
optimal deployment problem.

\vspace{5pt}
\noindent
\emph{Phase 1}
As proven in \cite{concur15}, the constraints in  
\ref{constraint-require}~$\dots$~\ref{constraint-distinct2} can be linearized.
Due to the fact that the remaining constraints \ref{bin-packing1} \dots 
\ref{bin-packing4} are the standard linear constraints of the bin 
packing problem, all the constraints of the phase 1 are linear and therefore 
the problem is in NP.
The hardness can be proven by reducing the bin packing problem to the 
considered problem. The reduction is straightforward: bins corresponds to 
nodes, packages are represented by microservices. The size of a package is 
encoded in the resource consumption of the microservice. The problem of 
minimizing the number of bins is therefore translated into finding the 
minimal amount of nodes to deploy the given microservices. To require the 
deployment of all nodes a new dummy target component of size 0 may be 
introduced using strong required interface for requiring the deployment of all 
the 
other microservices.

\vspace{5pt}
\noindent
\emph{Phase 2}
As far as the decision version of the phase 2 problem is concerned, it is clear 
that it is in NP due to the linearity of the constraints
\ref{const_binding1}~$\dots$~\ref{const_binding4}. 

To prove the decidability of the 
deployment problem, it is not needed to optimize the bindings. This, however, 
may be useful to express preferences over bindings on the final configuration.
In the following we study the complexity of the problem when a metric is used 
to try to optimize the bindings. We restrict ourselves to consider only linear 
metrics.

When linear metric constraints are used, the problem becomes NP-hard. By 
choosing the right metric it is 
indeed possible to reduce the partition problem into the considered problem.

The partition problem, a well-known NP-complete problem, checks the existence 
of a partition of a set 
$S$ into two subsets $A$, $B$ such that the difference between the sum of 
elements in $A$ and the sum of elements in $B$ is 0.

This problem can be encoded by using i) a 
microservice $\ResTy{T}_i$ 
for every number $i \in S$, ii) two 
microservices $\ResTy{T}_A$ and $\ResTy{T}_B$ representing the two sets $A$ and 
$B$, and iii) a dummy target microservice that requires the deployment of all 
the others.
We can enforce all the microservices to be deployed only once by allowing them 
to provide and be in conflict with the same interface ($p_i$ for the
$\ResTy{T}_i$ microservices, $p_a$ for $\ResTy{T}_A$ and $p_b$ for $\ResTy{T}_b$).
Every $\ResTy{T}_i$ should provide interfaces $p$ and $q$ with a 
numerical constraint $\leq 1$. $\ResTy{T}_A$ and $\ResTy{T}_B$ should instead 
weak require the interface $p$ with numerical constraint $\geq 0$ and 
provide the interface $q$. The dummy target microservice should only require 
$|S| + 2$ interfaces $q$.
With this universe of microservices, it is possible to define a metric that 
weights with $i$ every connection 
between $\ResTy{T}_A$ 
and $\ResTy{T}_i$, with $-i$ every connection between  
$\ResTy{T}_B$ and $\ResTy{T}_i$. The original partition problem can be solved 
by checking if the sum of the weights is 0.

\vspace{5pt}
\noindent
\emph{Phase 3}
It is easy to see that the 3rd phase is polynomial: it simply follows from 
the polinomial complexity 
of the topological sort over the number of components to be 
deployed and the set of 
interfaces $\Ifaces(U)=\bigcup_{\ResTy{T} \in U} 
\dom(\ResTy{T}\GetSReq{}) \cup \dom(\ResTy{T}\GetWReq{}) \cup 
\dom(\ResTy{T}\GetProv{}) \cup \ResTy{T}\GetConf{}$.
\qed
\end{proof}

\section{Graphical representation of the initial configuration}
\label{app:configuration}

\begin{sidewaysfigure}
\includegraphics[width=\textwidth]{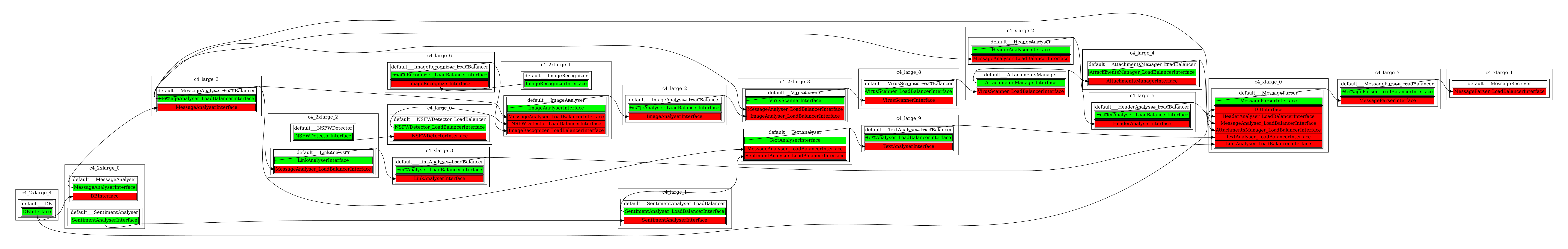}
\caption{Initial configuration of the email microservice system.}
\label{fig:initial_conf}
\end{sidewaysfigure}

\begin{figure}[tb]
	\centering
	\makebox[\textwidth][c]{\includegraphics[width=1.25\textwidth]{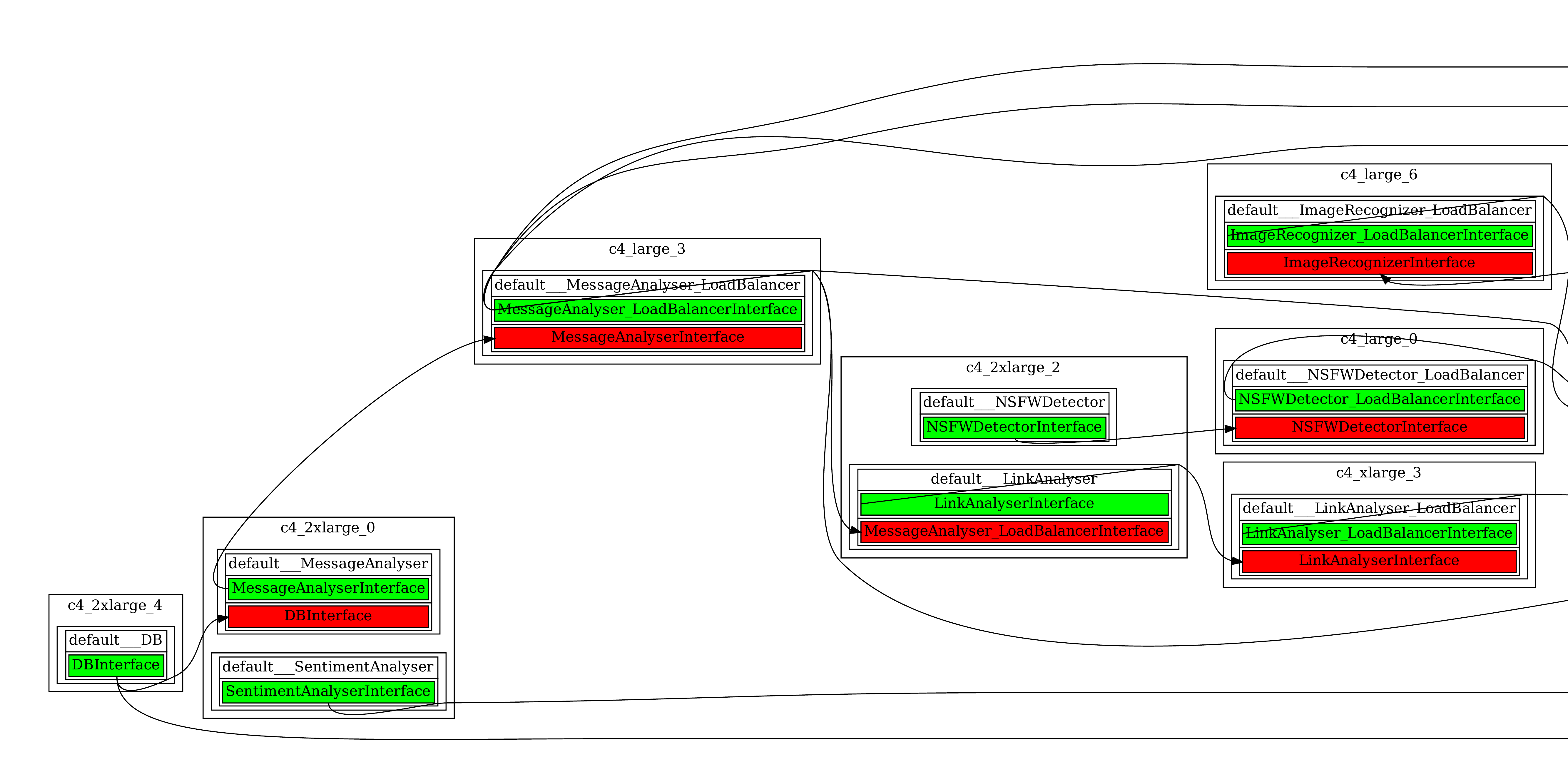}}
	\makebox[\textwidth][c]{\includegraphics[width=1.25\textwidth]{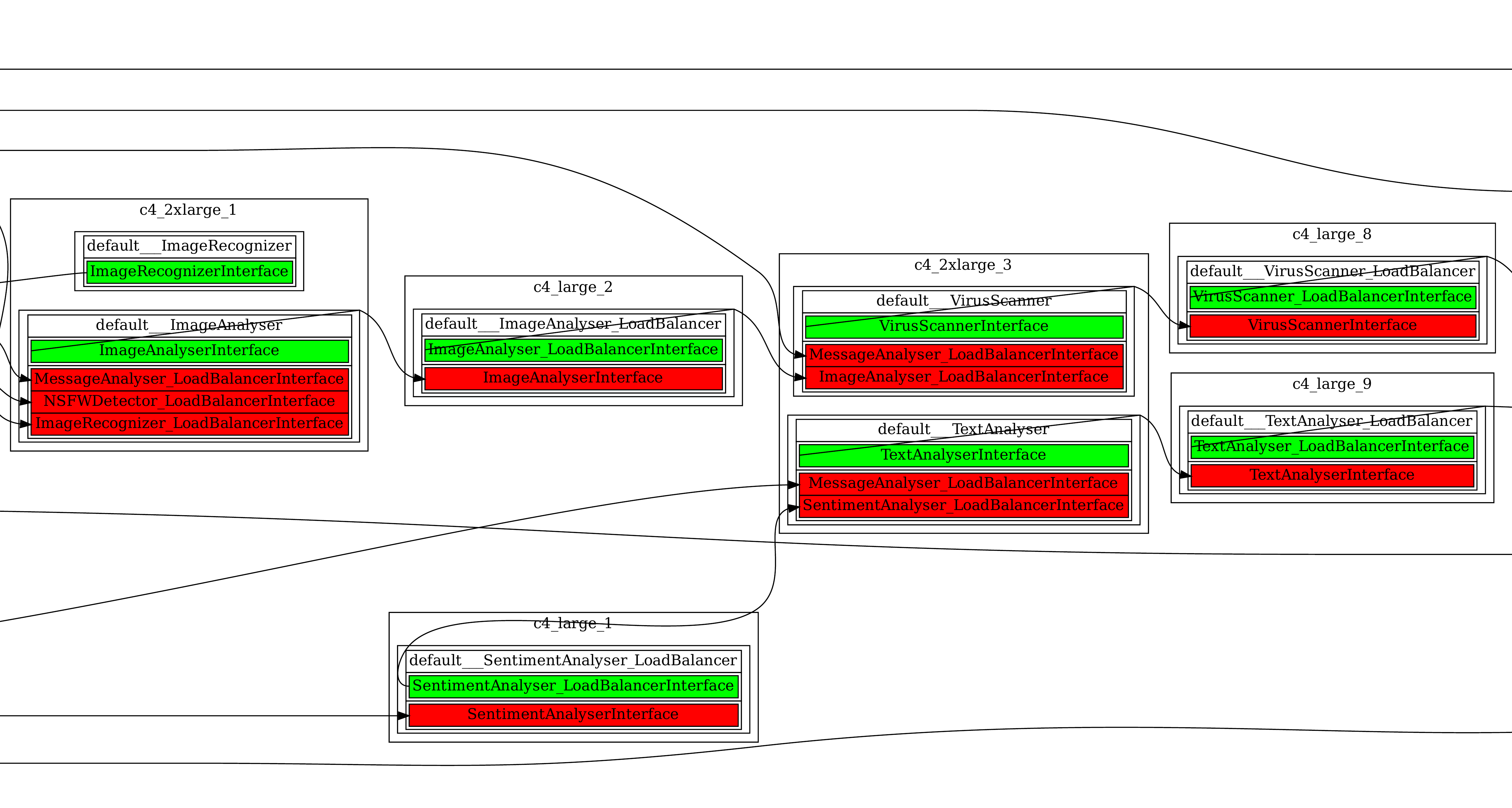}}
	\makebox[\textwidth][c]{\includegraphics[width=1.25\textwidth]{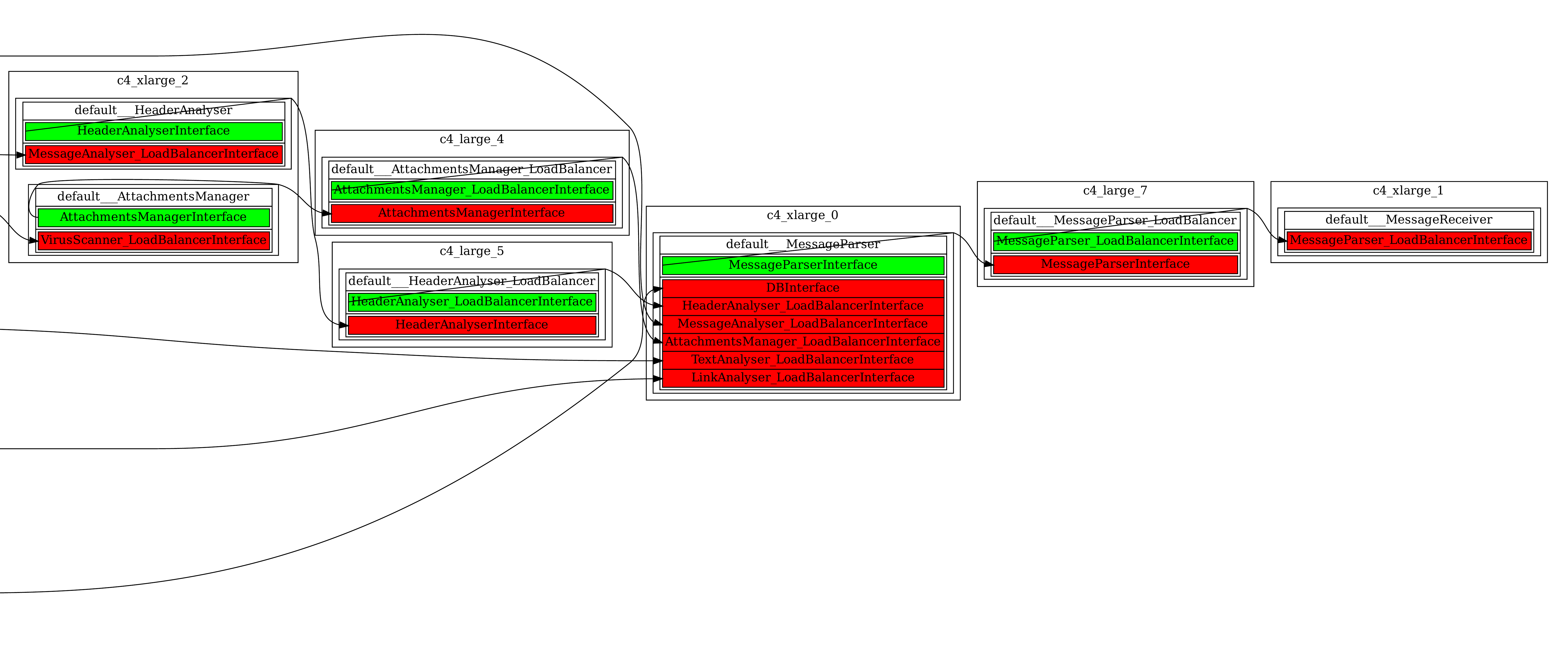}}
	\caption{Initial configuration of the email microservice system spitted in three for visualization purposes.}
	\label{fig:initial_conf_split}
\end{figure}

Figure \ref{fig:initial_conf} provides the graphical representation of the 
automatically synthesized initial 
configuration for our case study.
The same image, for visualization purposes, has been splitted in three and 
shown in Figure \ref{fig:initial_conf_split}.

In this figure, the outermost 
boxes represent the AWS virtual machines, while the innermost boxes represent 
the services deployed on that virtual machines.
The box names represent the kind of virtual machines 
used and the kind of objects deployed
(preceded by the word {\sf default}, corresponding
to an ABS/\tool parameter that we have not used
in our case-study).

The red boxes within a microservice $A$ represent the required
interfaces (either strong or weak), the green boxes represent the
provided interfaces of $A$.
An arrow 
from a service $A$ towards a service $B$ represents the fact that $A$ is used 
at runtime by $B$ and that $B$ needs to know the reference to $A$.

As can be seen from the image, the optimal initial deployment
consists of 24 
components, distributed over 5 virtual machines 
of type 2xlarge, 4 of type xlarge, and 10 of type large.

\end{document}